\documentclass[12pt]{article}
\usepackage[utf8]{inputenc}
\usepackage{amsmath}
\usepackage{amssymb}
\usepackage{amsthm}
\usepackage{algorithm,algpseudocode}
\usepackage{mathtools}
\usepackage{stmaryrd}
\usepackage{bm}
\usepackage{graphicx}
\usepackage{color}
\usepackage{hyperref}
\usepackage[anythingbreaks]{breakurl}
\renewcommand{\Pr}{{\mathbb{P}}}

\usepackage[margin=1.0in]{geometry}
\usepackage{cleveref}
\usepackage{setspace}
\newcounter{algsubstate}
\renewcommand{\thealgsubstate}{\alph{algsubstate}}

\newtheorem{theorem}{Theorem}

\newcommand{\calG}{\ensuremath{\mathcal{G}}}

\newcommand{\bone}{\mathbf{1}}
\newcommand{\bsigma}{\bm{\sigma}}

\newcommand{\bx}{\mathbf{x}}
\newcommand{\bw}{\mathbf{w}}
\newcommand{\cX}{\mathcal{X}}
\newcommand{\cU}{\mathcal{U}}
\newcommand{\by}{\mathbf{y}}
\newcommand{\bX}{\mathbf{X}}

\newcommand{\bE}{\mathbb{E}}
\newcommand{\beq}{\begin{equation}}
\newcommand{\eeq}{\end{equation}}
\newcommand{\blind}{0}

\usepackage{natbib}
\date{}

\begin{document}
\if0\blind
{
  \title{\bf Fast Conservative Monte Carlo Confidence Sets}
  \author{Amanda K.\ Glazer\thanks{
    The authors gratefully acknowledge NSF Grants DGE 2243822 and SaTC 2228884.}\hspace{.2cm}\\
    Department of Statistics, University of California\\
    and \\
    Philip B.\ Stark \\
    Department of Statistics, University of California}
  \maketitle
} \fi

\if1\blind
{
  \bigskip
  \bigskip
  \bigskip
  \begin{center}
    {\LARGE\bf Title}
\end{center}
  \medskip
} \fi

\bigskip

\newpage

\begin{abstract}
    Extant ``fast'' algorithms for Monte Carlo confidence sets are limited to univariate shift parameters
    for the one-sample and two-sample problems using the sample mean as the test statistic; moreover, some do not converge reliably and most do not produce conservative confidence sets.
    We outline general methods for constructing confidence sets for real-valued and multidimensional parameters 
    by inverting Monte Carlo tests using any test statistic and a broad range of randomization schemes.
    The method exploits two facts that, to our knowledge, had not been combined: (i)~there are Monte Carlo tests that are conservative despite relying on simulation, and (ii)~since the coverage probability of confidence sets depends only on the significance level of the test of the \emph{true}
    null, every null can be tested using the same
    Monte Carlo sample. 
    The Monte Carlo sample can be arbitrarily small, although the highest nontrivial attainable confidence level generally increases as the number $N$ of Monte Carlo replicates increases.
    We present open-source Python and R implementations of new algorithms to compute conservative confidence sets for real-valued parameters from Monte Carlo tests, for test statistics and randomization schemes that yield $P$-values that are monotone or weakly unimodal in the parameter, with the data and Monte Carlo sample held fixed.
    In this case, the new method finds conservative confidence sets for real-valued parameters in $O(n)$ time, where
    $n$ is the number of data.
    The values of some test statistics for different simulations and parameter values have a simple relationship that makes more savings possible.
\end{abstract}

\noindent%
\vfill

\section{Introduction}

While it is widely thought that tests based on Monte Carlo methods are approximate, it has been known for almost 90~years 
that Monte Carlo methods
can be used to construct conservative tests.
In particular, Monte Carlo $P$-values can be defined in a conservative way for:\footnote{%
This terminology is not universal and usage has changed over time; see, e.g., \citet{hemerik2024}.
}
\begin{itemize}
    \item \emph{simulation tests}, which sample from the (known) null distribution or from a `proposal distribution' with a known relationship to the null distribution
\citep{barnard1963,birnbaum1974, bolviken1996,davison1997, foutz1980, harrison12}

    \item \emph{random permutation tests}, which sample from the orbit of the observed data under the action of a group of transformations under which the null distribution is invariant
    \citep{davison1997, dwass1957, hemerik2021, pitman1937, ramdas2023}
    
    \item \emph{randomization tests}, which sample 
    re-randomizations of the observed data.
    \citep{hemerik2021, kempthorne1969, phipson2010, ramdas2023, zhang2023}
\end{itemize}
Simulation tests require the null distribution to be known; they sample from that distribution.
Random permutation tests require that the null satisfy a group
invariance; they sample from the null conditional on the orbit of the observed data under that group or a subgroup.
Randomization tests condition on the set of observed values and require that the data arise from (or be modeled as arising from) randomizing units into notional `treatments' following a known design.
Some tests can be thought of as either random permutation tests or randomization tests, for instance, the test based on the difference in sample means in the standard two-sample problem.

These tests involve auxiliary randomness in addition to the data: the randomness of the Monte Carlo 
method, which we denote generically by
$U$, a random variable with a known distribution on 
the measurable set $\mathcal{U}$.
It has been noted---but is not widely recognized---that to construct confidence sets by inverting hypothesis tests,\footnote{%
    See Theorem~\ref{thm:duality}.
}
the same Monte Carlo
simulations (the same realization of $U$) can be used to test every null \citep{harrison12}.
That is, suppose the data $X \in \mathcal{X}$ come from some 
unknown distribution $\Pr_\theta$ in the family of probability distributions (hypotheses) 
$$
\mathcal{P} := \{\Pr_\eta: \eta \in \Theta\}.
$$
We assume that $X$ and $U$ have a joint distribution;
$\Pr_{\eta, U}$ denotes that joint distribution when $X \sim \Pr_\eta$.
Let $P_\eta: \mathcal{X} \times \mathcal{U} \rightarrow [0, 1]$ for each $\eta \in \Theta$.
Suppose that 
for each $\eta \in \Theta$, 
$$
\Pr_{\eta, U} \{ P_\eta(X, U) \le p \} \le p \;\; \mbox{ for all } \;\; p \in [0, 1],
$$
i.e., that $\{ P_\eta: \eta \in \Theta\}$ is a family of conservative randomized $P$-values for the hypotheses $\theta = \eta$, $\eta \in \Theta$.\footnote{%
    The $P$-value $P_\eta(X, U)$
    is \emph{conservative} if
    $\Pr_{\eta, U} \{ P_\eta(X, U) \le p \} \le p$ for all $p \in [0, 1]$.
}
Then 
$$ S(X, U) := \{ \eta \in \Theta : P_\eta(X, U) \ge \alpha \}
$$ 
is a $1-\alpha$ confidence set for $\theta$: 
for every $\eta \in \Theta$,
$P$-value $P_\eta$ is calculated using the same value of $X$ and the same value of $U$.

Putting these two ideas together gives a computationally efficient, easily understood procedure for
constructing conservative confidence sets for univariate and multivariate parameters.\footnote{%
    A nominal confidence level $1-\alpha$ confidence procedure is \emph{conservative} if the chance it produces a set that contains the true parameter value is at least $1-\alpha$.
}
Suppose that $\Theta$ is a (possibly infinite) interval of real numbers.
If $P_\eta(x, u)$, viewed as a function of $\eta$ (with $x$ and $u$ fixed), is monotone, a one-sided confidence bound for $\theta$ can be found efficiently using a modified bisection
search.
If $P_\eta(x, u)$ is quasiconcave in $\eta$, a confidence interval for $\theta$ can be found efficiently 
using two modified bisection searches.
(We discuss computational strategies for non-quasiconcave $P$-values and for multivariate parameters in section~\ref{sec:confsets} below.)

Surprisingly, many texts that focus on permutation tests do not mention that some random permutation tests are conservative.
Instead they treat random permutation $P$-values as an approximation to the non-randomized ``full group'' (or ``all possible assignments'') $P$-value 
corresponding to the entire orbit of the observed data under the group \citep{good2006, higgins2004, pesarin2010}.
That point of view has led to 
methods for inverting permutation tests
that are computationally inefficient and approximate 
\citep{bardelli2016, garthwaite1996, garthwaite2009, pagano1983, tritchler1984}
even though there are more efficient, provably
conservative methods, as discussed below.

Consider the two-sample shift problem \citep{lehmann2005}, in which a fixed set of units are randomly assigned to control or treatment by simple random sampling.
Under the null hypothesis that 
$\theta = \eta \in \Re$, if a unit is assigned to treatment, its response differs by $\eta$ from the response it would have had if it had been assigned to control.\footnote{%
    \citet{caughey2017} show that tests for a shift that is assumed to be equal for all units have an interpretation that does not require that assumption: they are tests for the maximum or minimum shift for all units.
    Moreover, calculations involving constant shifts can be used to make confidence bounds for percentiles of a shift that may differ across units.
}
Consider as a test statistic the mean response of the treatment group minus the mean response of the control group (or the absolute value of that difference, for a two-sided test).
Many published numerical methods for finding confidence sets for $\theta$ from random permutation tests assume that the $P$-value is a continuous function of the hypothesized shift $\eta$ and crosses $\alpha$ at exactly one value of $\eta$ (for one-sided intervals) or
exactly two values (for two-sided intervals) \citep{bardelli2016, garthwaite1996, garthwaite2009, pagano1983, tritchler1984}.
But in fact, the $P$-value is a step function of $\eta$, as illustrated in Figure~\ref{fig:two-sample-p-shift}.

\begin{figure}[!ht]
    \centering
    \includegraphics[scale = 0.45]{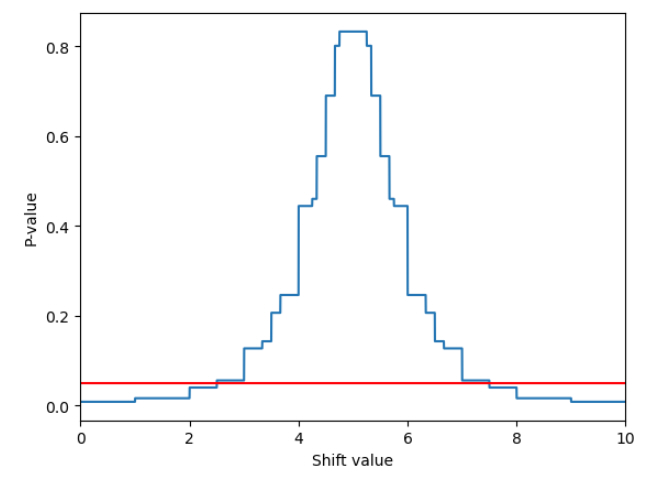}
    \caption{Exemplar (full-group) permutation $P$-value for the two-sided, two-sample shift problem as a function of the hypothesized shift $\eta$,
    using the absolute value of the difference in sample means as the test statistic.
    The two samples are $\{5, 6, 7, 8, 9\}$ and $\{0, 1, 2, 3, 4\}$.
    The $P$-value is piecewise constant, discontinuous, and
    quasiconcave.
    Many algorithms for inverting permutation tests to form confidence intervals incorrectly assume that the $P$-value is continuous in $\eta$ and equal to $\alpha$ at exactly two values of $\eta$.
    In general, the two-sided $P$-value for the two-sample shift problem using the difference in sample means (i)~is equal to $\alpha$ on two intervals, (ii)~is equal to $\alpha$ on one interval and jumps through $\alpha$ at a discontinuity, or (iii)~jumps through $\alpha$ at two discontinuities.
    The red horizontal line intersects the plot of the $P$-value as a function of the shift at values where the $P$-value equals 0.05:
    the $P$-value jumps through 0.05 at two discontinuities.}
    \label{fig:two-sample-p-shift}
\end{figure}

For some test statistics, the values of the test statistic
for different permutations of the data and for different values of $\eta$ have a simple relationship that makes even more computational savings possible by obviating the need to compute the test statistic from scratch for each Monte Carlo iteration and for each $\eta \in \Theta$.

This paper presents 
a general approach to finding conservative Monte Carlo
confidence sets using any test statistic and a broad range of randomization schemes (any for which there exists an exact test), flexibility that has not been addressed in the literature to our knowledge.
It also presents new algorithms 
and open-source Python and R implementations to compute confidence sets for scalar parameters from conservative Monte Carlo tests
when the $P$-value (for fixed data and a fixed set of Monte Carlo samples) is a monotone or weakly unimodal function of the parameter.
The new algorithms are more efficient than previous methods for computing approximate or conservative confidence sets from 
Monte Carlo tests
for the univariate one-sample and two-sample problems.

This paper is organized as follows.
Section~\ref{sec:extant-methods} reviews other ``fast'' methods for permutation confidence sets.
Section~\ref{sec:randomized-tests} gives an overview of randomized tests, $P$-values, and the duality between tests and confidence sets, highlighting
the fact that \emph{a confidence set 
derived by inverting hypothesis tests is conservative if and only if the test of the true null is conservative}.
Section~\ref{sec:permute} reviews some Monte Carlo tests that are conservative despite relying on simulation.
The general approach to computing confidence sets developed below can be used with any of them.
Section~\ref{sec:invert} discusses strategies to invert these tests to construct confidence sets and section~\ref{sec:confsets} presents the new algorithms to construct confidence sets for real-valued parameters when the $P$-value is monotone (for one-sided confidence bounds) or quasiconcave (for two-sided confidence intervals) in the hypothesized value of the parameter, using bisection  with a slight modification to make it conservative.
Section~\ref{sec:oneTwoSample} improves that result for two problems with additional structure: finding a confidence interval for the shift in one-sample and two-sample problems with real-valued data when the
test statistic is the sample mean (for the one-sample problem) or difference in sample means (for the two-sample problem).
Section~\ref{sec:numerical-comparisons} compares run times and confidence bounds for the new approach with those of
some other methods for the univariate one-sample and two-sample problems.
Section~\ref{sec:discuss} discusses extensions (including multi-dimensional confidence sets) and provides additional 
context.

\subsection{Real-valued parameters: extant methods}
\label{sec:extant-methods}
Many methods for Monte Carlo confidence sets seek to calculate or approximate full-group (or all-possible-assignments) $P$-values $P_\eta(X)$ (no dependence on $U$) for a collection of values of $\eta \in \Theta \subset \Re$.
Approximation methods (e.g., \citet{garthwaite1996, bardelli2016}) are not necessarily conservative.
Moreover, convergence guarantees for those algorithms require the $P$-value
to be continuous in $\eta$ and to cross $\alpha$ at two points \citep{tritchler1984, garthwaite1996, garthwaite2009, bardelli2016},
while in practice the $P$-value is typically a step function
(see Figure~\ref{fig:two-sample-p-shift}).

\citet{bardelli2016} use pseudorandom permutations to estimate the full-group $P$-value for each hypothesized value of the parameter in a bisection search
for the values of $\eta$ at which the $P$-value crosses $\alpha$. 
Each point in the search uses a new set of Monte Carlo simulations.
The standard bisection algorithm assumes that
the function is continuous and monotonic and
crosses $\alpha$ exactly once 
within each search interval.
Applying the standard bisection algorithm
as \citet{bardelli2016} do can fail:
(i)~the full-group $P$-value might not attain the target level for any $\eta$, or might attain the target level on an interval of values of $\eta$
 (figure~\ref{fig:two-sample-p-shift});
(ii)~sampling variability can make the simulated $P$-values non-monotonic in the shift; and
(iii)~re-estimating the $P$-value at a given $\eta$ will typically give a different value.

The first issue may cause bisection to terminate at the wrong root when $P_{\eta} = \alpha$ for a range of values of $\eta$, resulting in 
confidence intervals that are too short:
the lower endpoint of a confidence interval for $\theta$ is the 
\emph{smallest} $\eta$ for which $P_{\eta}
\ge \alpha$, and the upper endpoint of a confidence interval is the \emph{largest} $\eta$ for which $P_{\eta}
\ge \alpha$.
The second and third issues may cause the bisection algorithm to diverge.

There have been two basic approaches to reduce the computational burden of finding permutation confidence sets: 
speed up the calculation of exact full-group $P$-values using Fourier methods (e.g., \cite{tritchler1984, pagano1983}), or use stochastic approximations of the full-group $P$-value (e.g., \cite{garthwaite1996,garthwaite2009,ogorman2014}.
The latter approach can be used with Monte Carlo methods more generally, not just permutation methods.
Some methods approximate the full-group $P$-value by simulation, then account for the uncertainty in that approximation \citep{ernst2004permutation, good2006}.

To the best of our knowledge, extant algorithms only address two specific cases: the one-sample problem with the sample mean test statistic, or the two-sample problem with the difference-in-sample-means test statistic. 
We are not aware of any published algorithm that can handle general randomization schemes or that works with a variety of test statistics.

\citet{tritchler1984} 
seeks to find the full-group $P$-value efficiently by
taking advantage of the fact that the probability distribution of a sum is the convolution of the probability distributions, which becomes a product in the Fourier domain. 
The proof that their algorithm works relies on 
Theorem~1 of \citet{hartigan1969}
regarding ``typical values,'' which in turn assumes that the data have a continuous distribution---a
false assumption in many applications.
The method computes the distribution of the sum efficiently but evaluates 
that distribution at so many
points that in practice it is limited to small problems.
Unfortunately, as \citet{bardelli2016} notes, there seems to be no public implementation of this method and the description in
\citet{tritchler1984} omits essential details.

\citet{garthwaite1996}  
applies the Robbins-Monro stochastic optimization algorithm \citep{robbinsMonro51} to find an approximate confidence interval.
The Robbins-Monro algorithm is a stochastic iterative method for approximating a root of a univariate function of the expectation of a random variable from realizations of the random variable.
It assumes that the function is strictly monotone at the root, which is assumed to be unique; as previously mentioned, both assumptions are incorrect here (see figure~\ref{fig:two-sample-p-shift}).

The \citet{garthwaite1996} algorithm can be quite sensitive to the starting point;
\citet{garthwaite2009} propose averaging the results of the last $n^*$ iterations of the algorithm to increase stability. 
\citet{ogorman2014} shows that the \citet{garthwaite1996} and \citet{garthwaite2009} algorithms can produce very different confidence bounds---even from the same starting point---and recommend averaging eight runs. 
This increases the computational burden but 
does not guarantee that the resulting interval is
conservative.

\section{Randomized tests}
\label{sec:randomized-tests}
Rather than trying to approximate the ``full group'' or ``all assignments'' $P$-value, one can base confidence sets on conservative randomized tests.
Suppose we have a family of probability distributions 
$\{\mathbb{P}_\eta: \eta \in \Theta\}$ on the measurable space $\cX$, indexed by the abstract parameter $\eta \in \Theta$.
We will observe data $X \sim \Pr_\theta$ for some $\theta \in \Theta$.
We want to make tests and confidence sets for $\theta$
using $X$, possibly relying in addition on an independent,
auxiliary random variable $U$.\footnote{%
We treat the auxiliary randomness $U$ abstractly: it is not necessarily a uniformly distributed real-valued random variable,
as is common in the Neyman-Pearson framework.}
In the examples below, $U$ comprises 
randomness arising from Monte Carlo sampling.

For each $\eta \in \Theta$, let $\phi_\eta(x, u): (\cX,  \cU) \rightarrow \{0, 1\}$ be the
rejection function for a 
level $\alpha$ test of the hypothesis $\theta = \eta$:
reject the hypothesis $\theta = \eta$ 
when the data are $X=x$
and the auxiliary randomness is $U=u$ iff $\phi_\eta(x, u) = 1$.\footnote{%
    We assume that $\{\Pr_\eta\}_{\eta \in \Theta}$ have a common dominating measure and that $\phi_\eta$ is jointly measurable with respect to $\cX$ and $\cU$.
}
The test defined by $\phi_\eta$ has (conservative) significance level $\alpha$ for the hypothesis $\theta = \eta$ iff
\begin{equation} \label{eq:alpha-def}
    \bE_\eta \phi_\eta (X, U) \le
    \alpha,
\end{equation}
where the expectation is with respect to the joint distribution of $X$ and $U$,
computed on the assumption that $\theta = \eta$.

The duality between tests and confidence sets establishes that the set of all $\eta \in \Theta$ for which the hypothesis $X \sim \mathbb{P}_\eta$ would not be rejected at level $\alpha$ is a 
$1-\alpha$ confidence set for $\theta$:

\begin{theorem}[see, e.g., \citet{lehmann2005}, \S~3.5] \label{thm:duality}
For each $\eta \in \Theta$, let $\phi_\eta(x, u): (\cX, \cU) \rightarrow \{0, 1\}$ 
be the
rejection function for a test of the hypothesis $\theta = \eta$ at significance level $\alpha$.
Define $S(X, U) := \{ \eta \in \Theta : \phi_\eta(X, U) = 0 \}$. 
Then if the realized value of the data $X$ is $x$ and the realized value of the auxiliary randomness $U$ is $u$, $S(x, u)$ is a (possibly randomized) $1-\alpha$ confidence set for $\theta$.
\end{theorem}

\begin{proof}
\begin{equation}
    \Pr_{\theta,U} \{ S(X, U) \ni \theta \} =
    \Pr_{\theta, U} \{ \phi_\theta(X, U) < 1 \} 
    \ge 1-\alpha,
\end{equation}
where $\Pr_{\theta,U}$ is  the joint distribution of $X$ and $U$.
\end{proof}

\paragraph*{Remark~1}
The proof of Theorem~\ref{thm:duality} shows that the 
coverage probability of the confidence set rides entirely on the fact that the (single) test of the true null $\eta = \theta$ has level $\alpha$.
The tests involving other values of $\eta \in \Theta$ play no role whatsoever.
In particular, the tests of different nulls do not need to be valid simultaneously; dependence among them does not matter; and
a single $U$ can be used for every test.
Hence, if the Monte Carlo test of the \emph{true} null $\theta = \eta$
is conservative,
a single set of simulations (i.e., a single realization of $U$) can be used to test every $\eta \in \Theta$.

Section~\ref{sec:random-permutation} gives examples
of conservative Monte Carlo tests.
Many of the tests we consider
are derived from $P$-values, random variables with distributions that are stochastically dominated by the uniform distribution on $[0, 1]$ when the
null hypothesis is true:
$P_\eta = P_\eta(X, U)$ is a $P$-value for the hypothesis $\theta = \eta$ iff
\begin{equation} \label{eq:p-def}
  \Pr_{\eta, U} \{ P_\eta(X, U) \le p \} \le p,
  \;\;\; \forall p \in [0, 1], 
\end{equation}
where the
probability is with respect to the joint distribution $X$ and $U$.
If $P_\eta$ is a $P$-value, then (in the previous notation)
$$
\phi_\eta(x, u) := \mathbf{1}_{P_\eta(x, u) \le \alpha}$$
defines a test with significance level $\alpha$.

In turn, many of the $P$-values we consider arise from
a \emph{test statistic} $T : \mathcal{X} \rightarrow \Re$.
Monte Carlo simulation can be used to calibrate $T$ to produce conservative $P$-values, as discussed in the next section.

\section{Conservative Monte Carlo Tests} \label{sec:permute}

We begin by listing some conservative Monte Carlo tests.
All depend on a test statistic $T : \mathcal{X} \rightarrow \Re$, with larger values of $T$ considered to be
stronger evidence against the null, i.e., the $P$-value
decreases monotonically with $T$.
All use Monte Carlo to simulate $N$ new data sets $\{Y_j\}_{j=1}^N$ using the assumption that the null hypothesis is true.
They construct a $P$-value from the values of the test statistic for the original and simulated data.
The tests thus involve $X$, $T$, and $\{Y_j\}_{j=1}^N$.

In \emph{simulation tests}, $\{Y_j\}_{j=1}^N$ are simulated directly from the null distribution (possibly using
importance sampling).
In \emph{permutation tests}, $\{Y_j\}_{j=1}^N$ are 
generated by applying randomly selected elements of 
a group to the original data---a group under which
the probability distribution of the data is invariant if the null is true.
In \emph{randomization tests}, $\{Y_j\}_{j=1}^N$
are generated by randomly 
re-assigning the units
using the randomization design that was used
to collect the original data and adjusting the data
for each re-assignment 
using the assumption that the null is true.

\subsection{Simulation tests}
\label{sec:simulation-tests}
We observe data $X \sim \mathbb{P}$.
Consider the null hypothesis that $\mathbb{P} = \mathbb{P}^*$, a known distribution. 
Let $\{Y_j\}_{j=1}^n$ be IID $\mathbb{P}^*$.
Then the following is a conservative $P$-value \citep{barnard1963, birnbaum1974, bolviken1996, foutz1980}:
\beq \label{eq:randomized-no-weights}
    P := \frac{1 + \sum_j 1\{T(Y_j) \ge T(X) \}}{1+N}.
\eeq
The proof relies only on the fact that, if the null hypothesis is true, all rank orders of
$\{T(X), T(Y_1), \ldots, T(Y_N)\}$ are equally likely.
This result can be extended to sampling $\{Y_j\}$ from a known ``proposal distribution'' $\mathbb{Q}$ instead of sampling from the null distribution $\mathbb{P}^*$.
Suppose that $\mathbb{P}^*$ and $\mathbb{Q}$
are absolutely continuous with respect to a dominating measure 
$\mathbb{F}$ and that 
$d\mathbb{P}^*/d\mathbb{F}(x) = 0$ 
whenever $d\mathbb{Q}/d\mathbb{F}(x) = 0$
(i.e., any set with strictly positive probability under $\mathbb{P}^*$ has strictly positive probability under $\mathbb{Q}$). 
Let $\{Y_j\}_{j=1}^n$ be an IID sample from $\mathbb{Q}$. 
Define the \emph{importance weight}
$$ 
    w(x) := \frac{d\mathbb{P}^*/d\mathbb{F}}{d\mathbb{Q}/d\mathbb{F}}(x),
$$
with the convention $0/0 := 0$.
Then the following are conservative $P$-values \citep{harrison12}:

\beq \label{eq:randomized-weights-1}
    P := \frac{w(X) + \sum_{j=1}^N w(Y_j) 1\{ T(Y_j) \geq T(X) \}}{1 + N},
\eeq
and
\beq \label{eq:randomized-weights-2}
    P := \frac{w(X) + \sum_{j=1}^N w(Y_j) 1\{ T(Y_j) \geq T(X) \}}{w(X) + \sum_{j=1}^N w(Y_j)}.
\eeq
When $\mathbb{Q} = \mathbb{P}^*$,
$w \equiv 1$ and 
definitions~\ref{eq:randomized-weights-1}
and \ref{eq:randomized-weights-2} both reduce to 
equation~\ref{eq:randomized-no-weights}.

\subsection{Random permutation tests}
\label{sec:random-permutation}

Consider the null hypothesis that the probability distribution $\mathbb{P}$ of the data $X \in \mathcal{X}$ is invariant under some group $\calG$ of transformations on $\mathcal{X}$, so that $X \sim g(X)$ for all $g \in \calG$.
Random permutation tests involve generating $\{Y_j\}$ by applying randomly selected elements of $\calG$ 
to the original data $X$.
The elements can be selected in a number of ways:
 with or without replacement, with or without weights, from all of $\calG$ or from a subgroup of $\calG$; or 
 a randomly selected element of $\calG$ can be applied to a fixed subset of $\calG$. 
Below are some conservative $P$-values
for those sampling approaches.

\begin{enumerate}
    \item 
    $\{\hat{g}_i\}_{j=1}^N$ is a uniform random sample (with or without replacement) of size $N$ from $\calG$ or a subgroup of $\calG$.
    Let $Y_j := \hat{g}_j(X)$.
    Then

    \beq \label{eq:perm-P-1}
    P := \frac{1 + \sum_{j=1}^N 1\{T(Y_j) \geq T(X)\}}{1 + N}
    \eeq
    is a conservative $P$-value 
   \citep{dwass1957,davison1997,ramdas2023}.
    
    \item 
    $\{g_j\}_{j=1}^N$ is a fixed subset of 
    $N$ elements of $\calG$, not necessarily a subgroup; $\hat{g}$ is drawn uniformly at random from $\calG$; and $Y_j := g_j \hat{g}^{-1}(X)$.

    \beq \label{eq:perm-P-2}
    P := \frac{\sum_{j=1}^N 1\{T(Y_j)) \geq T(X) \}}{N}
    \eeq
    is a conservative $P$-value \citep{ramdas2023}.
    
    \item 
    $\calG$ is a finite group;
    $\hat{g}$ is selected at random from $\calG$ with probability $p(g)$ of selecting $g \in \calG$;
    and $Y_j := g_j \hat{g}^{-1}(X)$.
    
    \beq \label{eq:perm-unequal-prob}
    P :=  \sum_{j=1}^{|\calG|} p(g_j) \cdot 
    1\{ T(Y_j) \geq T(X) \}
    \eeq is a conservative $P$-value \citep{ramdas2023}.
    
\end{enumerate}

\noindent
Note that equation~\ref{eq:perm-P-1} has the same form as 
equation~\ref{eq:randomized-no-weights}, but equation~\ref{eq:perm-P-1}
in general involves sampling from the conditional distribution given the orbit of the observed data, while equations~\ref{eq:randomized-no-weights},
\ref{eq:randomized-weights-1},
and \ref{eq:randomized-weights-2} involve sampling from the unconditional distribution.
The $P$-values in equations~\ref{eq:perm-P-2}
and \ref{eq:perm-unequal-prob} involve
drawing only one random permutation $\hat{g}$, then applying it to other elements of $\cal{G}$ to create $\{Y_j\}_{j=1}^N$.

\subsection{Randomization tests}
\label{sec:randomization-tests}
Unlike (random) permutation tests, which rely on the invariance of the distribution $\mathbb{P}$ of $X$ under $\calG$
if the null is true, randomization tests rely on the fact that
generating the data involved randomizing units to treatments (in reality or hypothetically).

Suppose there are $n$ units, each of which is
assigned to one of $K$ possible treatments, $t=1, \ldots K$.
A \emph{treatment assignment} $\bsigma$ assigns a treatment to each unit: it is a mapping from $\{1, \ldots, n\}$ to 
$\Sigma \subset \{1, \ldots, K\}^n$,
where 
$\Sigma$ is the set of treatment assignments the original randomization design might have produced.
Depending on the design, $\Sigma$ will be a different subset of $\{1, \ldots, K\}^n$.
The assignment might use simple random sampling, Bernoulli sampling, blocking, or stratification;
the assignment probabilities might depend on
covariates.
Let 
$p(\bsigma)$ be the probability that the original randomization would make the assignment $\bsigma$, for $\bsigma \in \Sigma$. 
Let $\bsigma_0$ denote the actual treatment assignment.
The data $X := (W, \bsigma_0)$ comprise the (vector of) responses $W$
and the corresponding treatment assignment $\bsigma_0$.

Suppose we draw a weighted random sample of $N$ treatment assignments $\bsigma_j \in \{1, \ldots, K\}^n$, $j=1, \ldots, N$, with or without replacement, with probability $p(\bsigma)$ of selecting $\bsigma \in \Sigma$ in the first draw (adjusting the selection probabilities appropriately in subsequent draws if the sample is drawn without replacement).
Let $Y_j := (W, \bsigma_j)$, $j=0, \ldots, N$. 
Then
\beq \label{eq:rand-sample}
P := \frac{\sum_{j=0}^{N} p(\bsigma_j) \cdot 
    1\{ T(Y_j) \geq T(X) \}}{\sum_{j=0}^N p(\bsigma_j)}
\eeq 
is a conservative $P$-value.
Equation~\ref{eq:rand-sample} has the same form as equation~\ref{eq:randomized-weights-2}, and when all $\bsigma
\in \Sigma$ are equally likely, equation~\ref{eq:rand-sample} has the same form as equations~\ref{eq:randomized-no-weights} and \ref{eq:perm-P-1}.

Another conservative $P$-value 
has the same form as equation~\ref{eq:perm-unequal-prob}; it involves a weighted sum over all possible assignments
$\Sigma$, just as equation~\ref{eq:perm-unequal-prob}
involves a weighted sum over all group elements.
Let $Y_j := (W, \bsigma_j)$.
Then
\beq \label{eq:rand}
P := \sum_{j=1}^{|\Sigma|} p(\bsigma_j) \cdot 
    1\{ T(Y_j) \geq T(X) \}
\eeq
is a conservative $P$-value 
\citep{hemerik2021, kempthorne1969, zhang2023, ramdas2023}.

\subsection{Tests about parameters}
\label{sec:parameter}
The tests above do not explicitly involve parameters.
They can be extended to a family of tests, one for
each possible parameter value $\eta \in \Theta$,
in various ways.
For example, suppose there is a bijection $f_\eta: \mathcal{X} \rightarrow \mathcal{X}$
such that for any $\mathbb{P}_0$-measurable set
$A$, $f_\eta(A)$ is $\mathbb{P}_\eta$-measurable and $\mathbb{P}_\eta (f_\eta(A)) = \mathbb{P}_0 (A)$,
and for any $\mathbb{P}_\eta$-measurable set
$A$, $f_\eta^{-1}(A)$ is $\mathbb{P}_0$-measurable and $\mathbb{P}_\eta (A) = \mathbb{P}_0 (f_\eta^{-1}(A))$.
Then a $P$-value for the hypothesis $\theta = \eta$
can be calculated by using $T(f_\eta^{-1}(X))$ and 
$\{T(f_\eta^{-1}(Y_j))\}$ in place of of 
$T(X)$ and 
$\{T(Y_j)\}$ in equation~\ref{eq:randomized-no-weights}, \ref{eq:perm-P-1}, or \ref{eq:perm-P-2}.
For tests about the effect of treatment in the two-sample problem, see section~\ref{sec:oneTwoSample}.

Another common problem is to test a hypothesized value 
of a location parameter $\theta \in \Theta \subset \Re^m = \mathcal{X}$ in a location family
$\{\mathbb{P}_\eta: \eta \in \Theta\}$. 
In a location family, 
for any $\mathbb{P}_{\mathbf{0}}$-measurable set $A \in \mathcal{X}$, $\mathbb{P}_\eta (A) := \mathbb{P}_{\mathbf{0}} (A-\eta)$.
In other words, $f_\eta(x) = x + \eta$
and $f_\eta^{-1}(x) = x - \eta$.
A test of the hypothesis $\theta = \mathbf{0}$ can be used to test the hypothesis $\theta = \eta$ by applying the test to the transformed data $X-\eta$ and the transformed values $\{Y_j - \eta\}$.
For simulation test $P$-values using importance sampling (definitions \ref{eq:randomized-weights-1} and \ref{eq:randomized-weights-2}) we can test the hypothesis 
$\theta = \eta$ by using the weights $w_\eta(x) := d\mathbb{P}_\eta/d\mathbb{Q}$; nothing else needs to be changed.
In particular, the same sample can be used.

\section{Inverting tests to construct confidence sets}
\label{sec:invert}
We use $P_\eta(X, U)$ to denote a $P$-value for the hypothesis $\theta = \eta$, for data $X$ and auxiliary
randomness $U$.
We assume that the family $\{P_\eta(x, u): \eta \in \Theta, \;x \in \mathcal{X}, \; u \in \mathcal{U} \}$
of $P$-values has been selected before observing $X$ or $U$.
Remark~1, above, notes that it suffices to use a single set of Monte Carlo simulations (a single realization of $U$) to test the hypothesis 
$\theta = \eta$ for all 
$\eta \in \Theta$.
As mentioned in section~\ref{sec:parameter},
how the simulations can be re-used depends on how the
test depends on the parameter $\eta$.

Given the $P$-value function $P_\eta(\cdot, \cdot)$, a $1-\alpha$ confidence set for $\theta$ is the $\alpha$ superlevel set of the (deterministic) $P$-value function for the observed values $X=x$ and $U=u$:
$$
S_\alpha(x, u) := \{ \eta \in \theta : P_\eta(x, u) \ge 1- \alpha\}.
$$
How can we find that set (or a set that contains it, resulting in a more conservative confidence set) in practice?

The $P$-value function is bounded, but unless it has additional structure, it is in general impossible to calculate its $\alpha$ superlevel set.
The numerical problem amounts to finding a particular contour of a discontinuous function on an unbounded domain.

There are a number of strategies that may help find 
a valid confidence set, depending on details of the $P$-value function
and prior constraints on $\theta$:
\begin{enumerate}
    \item If $\Theta$ is a subset of $\Re$ and
    $P_\eta$ is monotone or weakly unimodal (i.e., quasiconcave), a conservative confidence set can be found by a modification of the bisection algorithm.
    See section~\ref{sec:confsets}.

    \item Even in dimensions greater than 1, if $P_\eta(x, u)$ is a weakly unimodal function of $\eta$,\footnote{%
    A function $f$ on $\Re^\ell$ is weakly unimodal if there is a point $\eta \in \Re^\ell$ such that for all $\zeta \in \Re^\ell$ with $\|\zeta\| = 1$, $f(\eta + c\zeta) \le f(\eta)$ for all $c \ge 0$.
    } 
    finding a superset of the confidence set may be computationally straightforward.
    In such situations, a suitable estimate of $\theta$ may be a mode of $P_\eta$.
    See section~\ref{sec:confsets}.

    \item Test every value
    of the parameter where the $P$-value could change.
    In some situations, the set of such values can be characterized explicitly (c.f.\ equation~\ref{eq:two_sample_shift}). 
    
    \item If $\Theta$ is a bounded subset of $\Re^\ell$ and the modulus of continuity of the $P$-value function has a known upper bound,
    one can test a grid of values of $\eta$
    and refine the grid adaptively to ensure that between grid points, the $P$-value must be on the same side of $\alpha$ as its nearest neighbor.

    \item If $\Theta$ is a bounded subset of $\Re$
    and the $P$-value function can be written as 
    difference of two monotone functions, 
    one can test a grid of values of $\eta$
    and refine the grid adaptively to ensure that between grid points, the $P$-value must be on the same side of $\alpha$ by ignoring one of the two functions, depending on whether the values at neighboring gridpoints are both above or both below $\alpha$ \citep{ottoboni2018}.

    \item In some situations, it may be possible to find extreme points of the (not necessarily connected) confidence set; the convex hull of those points is a conservative confidence set.
\end{enumerate}
Below we explore the case of scalar parameters and monotone or quasiconcave (weakly unimodal) $P$-value functions.

\subsection{Confidence intervals for scalar parameters when the \textit{P}-value is monotone or quasiconcave}
\label{sec:confsets}

\subsubsection{One-sided confidence bounds}
Suppose that $P_\eta(X, U)$ is a conservative 
$P$-value for the hypothesis $\theta = \eta$ for data $X$ and auxiliary randomness $U$, for instance, derived from a test statistic
$T_\eta$ using random permutations as described in Section~\ref{sec:random-permutation}.
Given data $X=x$ and auxiliary randomness $U=u$, 
a lower (upper) $1-\alpha$ confidence bound for the scalar parameter $\theta$ is the smallest (largest) value of $\eta$ for which $P_\eta(x, u) \ge \alpha$;
every $P$-value $P_\eta$ is calculated using the same data and the same set of random permutations, i.e., the same values of $X$ and $U$.
If $P_\eta(x, u)$ is monotone increasing (decreasing) in $\eta$, it is straightforward to find a lower (upper) confidence bound
using a modification of the bisection algorithm (to account for the fact that $P(\eta)$ is not continuous in general;
see, e.g., figure~\ref{fig:two-sample-p-shift} and
\cite{aronow2023}), as
shown in Algorithm~\ref{alg:one-sided-lower} and \ref{alg:one-sided-upper}.

Suppose $\theta^-$ is the largest conservative
one-sided lower confidence bound based on inverting tests defined by the monotonically nondecreasing $P$-value function $P_\eta(x, u)$, i.e., 
\beq
   \theta^- := \sup \{\eta : P_\zeta(x, u) \le \alpha \;\; \forall \zeta \le \eta \}.
\eeq
Algorithm~\ref{alg:one-sided-lower} returns a value  in $[\theta^--e, \theta^-]$
for any specified tolerance $e > 0$. 
Similarly, suppose $\theta^+$ is the smallest conservative one-sided upper confidence bound based on inverting tests defined by the monotonically nonincreasing $P$-value function $P_\eta(x, u)$, i.e., 
\beq
   \theta^+ := \inf \{\eta : P_\zeta(x, u) \le \alpha \;\; \forall \zeta \ge \eta \}
\eeq
Algorithm~\ref{alg:one-sided-upper} returns a value in $[\theta^+, \theta^+ +e]$ for any specified tolerance $e > 0$. 

\begin{algorithm} 
\caption{
\protect \label{alg:one-sided-lower} 
Modified Bisection Algorithm for Lower Confidence Bound (LCB) for $\theta$}
\begin{algorithmic}
  \small
  \Statex \textbf{Input:} Data $X=x$ and auxiliary randomness $U=u$; conservative $P$-value $P_\eta(X, U)$ such that $P_\eta(x, u)$ is monotonically nondecreasing with $\eta$;  non-coverage level $\alpha \in (0, 1)$; tolerance $e > 0$; 
  initial step size $\delta_0 > 0$ (e.g., $\delta_0 = 1)$; 
  initial trial value $\eta_0 \in \Re$ (e.g., $\eta_0 = 0$).
  \Statex \textbf{Output:} Conservative $1-\alpha$ LCB for $\theta$ within $e$ of the largest valid 
  $1-\alpha$ LCB based on 
  $P_{\eta}(x, u)$
  \Statex 
   \textbf{Step 1: Find finite a interval $[a, b]$ that brackets the LCB}
  \If{$P_{\eta_0}(x, u) \ge \alpha$} \Comment{$P_{\eta_0}(x, u) \ge \alpha$}
    \State $b \gets \eta_0$
    \State $\delta \gets \delta_0$
    \State $\eta_* \gets \eta_0 - \delta$
    \While{$P_{\eta_*}(x, u) \ge \alpha$} \Comment{Search for finite lower bound on the LCB}
        \State $\delta \gets 2\delta$
        \State $\eta_* \gets \eta_* - \delta$
    \EndWhile
    \State $a \gets \eta_*$
    
 \Else \Comment{$P_{\eta_0}(x, u) < \alpha$}  
    \State $a \gets \eta_0$
    \State $\delta \gets \delta_0$
    \State $\eta_* \gets \eta_0 + \delta$
    \While{$P_{\eta_*}(x, u) < \alpha$} \Comment{Search for finite upper bound on the LCB}
        \State $\delta \gets 2\delta$
        \State $\eta_* \gets \eta_* + \delta$
    \EndWhile
    \State $b \gets \eta_*$
  \EndIf

  \Statex \textbf{Step 2: Conservative bisection to find LCB}
  \While{$|b - a| > e$}
    \State $\eta_* \gets (a + b) / 2$ 
    \If{$P_{\eta_*}(x, u) < \alpha$}
      \State $a \gets \eta_*$
    \Else
      \State $b \gets \eta_*$
    \EndIf
  \EndWhile
  
  \State \textbf{Return:} $a$ \Comment{$a \in [\theta^- - e, \theta^-]$ is a conservative $1-\alpha$ lower confidence bound for $\theta$}
\end{algorithmic} 
\end{algorithm}

\begin{algorithm} 
\caption{
\protect \label{alg:one-sided-upper} 
Modified Bisection Algorithm for Upper Confidence Bound (UCB) for $\theta$}
\begin{algorithmic}
  \small
  \Statex \textbf{Input:} Data $X=x$ and auxiliary randomness $U=u$; conservative $P$-value $P_\eta(X, U)$ such that $P_\eta(x, u)$ is monotonically nonincreasing with $\eta$; non-coverage level $\alpha \in (0, 1)$; tolerance $e > 0$; 
  initial step size $\delta_0 > 0$ (e.g., $\delta_0 = 1)$; 
  initial trial value $\eta_0 \in \Re$ (e.g., $\eta_0 = 0$).
  \Statex \textbf{Output:} Conservative $1-\alpha$ UCB for $\theta$ within $e$ of the smallest valid 
  $1-\alpha$ UCB based on 
  $P_{\eta}(x, u)$
  \Statex 
   \textbf{Step 1: Find a finite interval $[a, b]$ that brackets the UCB}
  \If{$P_{\eta_0}(x, u) \ge \alpha$} \Comment{$P_{\eta_0}(x, u) \ge \alpha$}
    \State $a \gets \eta_0$
    \State $\delta \gets \delta_0$
    \State $\eta_* \gets \eta_0 + \delta$
    \While{$P_{\eta_*}(x, u) \ge \alpha$} \Comment{Search for finite upper bound on the UCB}
        \State $\delta \gets 2\delta$
        \State $\eta_* \gets \eta_* + \delta$
    \EndWhile
    \State $b \gets \eta_*$
    
 \Else \Comment{$P_{\eta_0}(x, u) < \alpha$}  
    \State $b \gets \eta_0$
    \State $\delta \gets \delta_0$
    \State $\eta_* \gets \eta_0 - \delta$
    \While{$P_{\eta_*}(x, u) < \alpha$} \Comment{Search for finite lower bound on the UCB}
        \State $\delta \gets 2\delta$
        \State $\eta_* \gets \eta_* - \delta$
    \EndWhile
    \State $a \gets \eta_*$
  \EndIf

  \Statex \textbf{Step 2: Conservative bisection to find UCB}
  \While{$|b - a| > e$}
    \State $\eta_* \gets (a + b) / 2$ 
    \If{$P_{\eta_*}(x, u) < \alpha$}
      \State $b \gets \eta_*$
    \Else
      \State $a \gets \eta_*$
    \EndIf
  \EndWhile
  
  \State \textbf{Return:} $b$ \Comment{$b \in [\theta^+, \theta^+ + e]$ is a conservative $1-\alpha$ upper confidence bound for $\theta$}
\end{algorithmic} 
\end{algorithm}

In some circumstances there are more efficient ways to pick $\eta_0$ and $\delta_0$ in Algorithm~\ref{alg:one-sided-lower} and \ref{alg:one-sided-upper}.
For instance, for the two-sample problem using the difference in sample means as the test statistic,
picking $\eta_0$ to be the difference in sample means for the original data and $\delta_0$ to be the range of the data may eliminate the need to search for initial values of $L$ and $U$.

\subsubsection{Two-sided confidence intervals}

Suppose that $P_\eta(x, u)$ is quasiconcave in $\eta$, i.e., that it is weakly unimodal.
Then the confidence set---all $\eta \in \Theta$ for which the hypothesis $\theta = \eta$ is not rejected at level $\alpha$---is an interval.
Its lower endpoint is the largest $\eta$ such that $P_\zeta(x, u) \le \alpha$ for all $\zeta \in (-\infty, \eta]$;
its upper endpoint $u$ is the smallest $\eta$ 
such that $P_\zeta(x, u) \le \alpha$ for all $\zeta \in [\eta, \infty)$.

For example, consider the
standard approach of combining two one-sided tests, each at half the significance level.
If both one-sided tests are conservative, then by Bonferroni's inequality, the resulting 
two-sided test is conservative even when the two $P$-values have arbitrary dependence.
That allows us to use the same set of random permutations to construct both one-sided $P$-values and thus the two-sided confidence interval.

If one of the $P$-values (say $P_\eta^1(x, u)$) is monotone nondecreasing in $\eta$ and the other ($P_\eta^2(x, u)$) is monotone nonincreasing in $\eta$, then the overall $P$-value function $2 \min \{P_\eta^1(x, u), P_\eta^2(x, u)\}$ is quasiconcave in $\eta$.
For example, the two-sample test using
the difference in sample means as the test statistic $T$ has this property (see Section~\ref{sec:oneTwoSample}).

Quasiconcavity of $P_\eta(x, u)$ in $\eta$
makes it efficient to find both endpoints of a confidence interval for $\theta$ using a minor modification of
the bisection method.
The corresponding procedure is given in Algorithm~\ref{alg:modified-bisect-2-sided}.

\begin{algorithm} 
\caption{Modified bisection algorithm for the two-sided confidence interval (CI) for $\theta$}
\begin{algorithmic}[1]
  \small
  \Statex \textbf{Input:} Data $X=x$ and auxiliary randomness $U=u$; conservative $P$-value $P_\eta(x, u)$ such that $P_\eta(x, u)$ is quasiconcave in $\eta$;  non-coverage level $\alpha \in (0, 1)$; tolerance $e > 0$; 
  initial step size $\delta_0 > 0$ (e.g., $\delta_0 = 1)$; 
  initial value $\eta_0 \in \Re$ such that $P_{\eta_0}(x, u) \ge \alpha$.
  \Statex \textbf{Output:} Conservative $1-\alpha$ CI whose endpoints are each within $e$ of the corresponding endpoints of the smallest valid 
  $1-\alpha$ CI based on 
  $P_{\eta}(x, u)$
  \Statex 
   \textbf{Step 1: Enter algorithm~\ref{alg:one-sided-lower} 
   in step~1 at the branch $P_{\eta_0} \ge \alpha$ to find the conservative lower confidence bound $a$}

  \Statex \textbf{Step 2: Enter algorithm~\ref{alg:one-sided-upper} 
   in step~1 at the branch $P_{\eta_0} \ge \alpha$ to find the conservative upper confidence bound $b$}
   
  \Statex \textbf{Return:} $[a, b]$ \Comment{$1-\alpha$ CI, no more than $2e$ longer than shortest valid $1-\alpha$ CI based on $P_\eta$}
\end{algorithmic} \label{alg:modified-bisect-2-sided}
\end{algorithm}

\footnotetext{
      A value of $L$ can generally by found using an
      estimate of $\theta$ from the data $X$.
      Given a value of $L$, a value of $U$ can be found by adding a sufficiently large value to $L$, since the $P$-value eventually decreases monotonically as $\eta$ increases.
    }

\section{Additional efficiency in the one-sample and two-sample shift problems}
\label{sec:oneTwoSample}

For the most common test statistics for the one-sample and two-sample problems, the computational cost of finding each $P$-value can be reduced further
by saving the Monte Carlo treatment assignments (or in some cases, 
just a one-number summary of each assignment: see equations~\ref{eq:one_sample_shift} and \ref{eq:two_sample_shift}) 
and the value of the test statistic for each treatment assignment \citep{nguyen2009, pitman1937}.
This section explains how.

\paragraph*{The one-sample problem.} 
The one-sample problem is to find a confidence interval for the center of symmetry $\theta$ of a symmetric distribution $\Pr_\theta$
from an IID sample $X = \{X_j \}_{j=1}^n$
from $\Pr_\theta$.
If $\theta = \eta$, the distribution of
$X_j - \eta$ is symmetric around 0, i.e.,
conditional on $|X_j - \eta|$,
$X_j - \eta$ is equally likely to be $\pm(X_j - \eta)$.
A common test statistic for the hypothesis $\theta = \eta$ is the sum of the signed differences from $\eta$:
\beq
  T_\eta(\boldsymbol{x}) := \sum_{j=1}^n (x_j-\eta).
\eeq
Random permutation tests for the one-sample problem involve the distribution of this test statistic when
the signs of $\{X_j - \eta\}$ are randomized independently.
Let $\bsigma \in \{-1, 1\}^n$ be IID uniform random signs.
For any two vectors $\bx$, $\by$, with the same dimension $n$, define the componentwise product $\bx \odot \by := (x_1 y_1, \ldots, x_n y_n)$.
Let $\bX$ denote an $n$-tuple comprising the data $\{X_j\}$ in some canonical order.
We can test the hypothesis $\theta = \eta$ by comparing the value of $T_\eta(\bX)$
to the values of $T_\eta(\bsigma \odot (\bX- \eta \bone ) + \eta\bone )$ 
for $N$ randomly generated sign vectors $\bsigma$.
As noted above in section~\ref{sec:confsets},
we can re-use the values of $\bsigma$ to test $\theta = \eta$ for different values of $\eta$.
We can save even more computation by relating the value of \begin{equation}
T_{\eta, \bsigma}(\bX) :=     
T_\eta(\bsigma \odot (\bX - \eta \bone) + \eta \bone)
\end{equation}
to the value of 
$T_{0, \bsigma}(\bX) = T_0(\bsigma \odot \bX)$:
\begin{eqnarray}
    T_\eta (\bsigma \odot (\bx - \eta \bone) + \eta \bone) &=&
    \sum_{j=1}^n \left ( [\sigma_j (x_j - \eta) + \eta ] - \eta \right  ) \nonumber \\
    &=& 
    \sum_{j=1}^n \sigma_j x_j - \sum_j \sigma_j \eta \nonumber \\
    &=& T_0(\bsigma \odot \bX) - \eta \sum_j \sigma_j. \label{eq:one_sample_shift}
\end{eqnarray}
Thus, for each vector of signs $\bsigma$, we need only keep track of $T_0(\bsigma \odot \bX )$ and $\sum_j \sigma_j$ to determine the value of the test statistic for any other hypothesized value of $\eta$.

\paragraph*{The two-sample problem.}
The two-sample problem involves a set of $n$ units each randomized to one of two conditions, treatment or control,
by randomly allocating $m$ units into treatment and the remaining $n-m$ to control by simple random sampling, with $n$ and $m$ fixed in advance.
The response of unit $j$ is 
$w_j = r_j + \sigma_j \theta$, where $r_j$ is the `baseline' response of unit $j$, $\sigma_j = 1$ if unit $j$ is assigned to treatment, and $\sigma_j = 0$ if unit $j$ is assigned to control.\footnote{%
    This is the \emph{the Neyman model for causal inference} \citep{neyman1990, imbens2015causal}. 
    The Neyman model implicitly assumes that
    each unit's response depends only on whether that unit
    is assigned to treatment or control, not on the assignment of other units.
    \citet{caughey2017} give conditions under which this model (with particular test statistics) can be used to draw inferences about percentiles of the effect of treatment even when the effect varies across units, i.e., when
    $w_j = r_j + \sigma_j \theta_j$.
}
We seek a confidence interval for $\theta$ from the resulting data.
Let $\mathbf{W} := (W_1, \ldots, W_n)$ be an $n$-tuple comprising the
responses in some canonical order in which $W_1, \ldots, W_m$
are the responses of the $m$ units assigned to treatment and
$W_{m+1}, \ldots, W_n$ are the responses of the $n-m$ units assigned to control.
Let $\bsigma$ be a vector containing $m$ ones and $(n-m)$ zeros, representing a treatment assignment.
(In the original assignment $\bsigma_0$, the first $m$ elements are 1 and the remaining $n-m$ are 0.)

Randomization tests for the two-sample problem involve the distribution of the test statistic when
units are randomly assigned to treatment and control using the same randomization design the original experiment used,
on the assumption that the shift is $\eta$, i.e., that if a unit 
originally assigned to treatment
had instead been assigned to control, its response
would have been less by $\eta$, and if a unit that was originally assigned to control
had instead been assigned to treatment,
its response would have been greater by $\eta$.

Recall from section~\ref{sec:randomization-tests} that for randomization tests, the data $x$ comprise an ordered pair: the vector $\bw$ of responses and the corresponding vector $\bsigma$ of treatment assignments.
A common test statistic is the difference between
the mean response of the treatment group and the mean
response of the control group,\footnote{%
    A test that uses the difference in sample means as the test statistic is
    equivalent to a test based on the mean of either group, since the total is fixed: the statistics are monotone functions of each other.
} 
\begin{equation}
T_\eta(x) = T_\eta((\bw, \bsigma)) := \frac{1}{m}\sum_{j:\sigma_j = 1} (w_j + \eta1_{j > m}) - \frac{1}{n-m}\sum_{j:\sigma_j = 0} (w_j - \eta1_{j \le m}),
\end{equation}
or its Studentized
version \citep{wu2021}.
The term $+\;\eta 1_{j>m}$ in the first summand adds the hypothesized treatment effect to units that were originally assigned to
control ($j>m$) but are re-assigned to treatment ($\sigma_j = 1$);
the term $-\;\eta 1_{j\le m}$ in the second summand subtracts the 
hypothesized treatment effect from units that were originally assigned to
treatment ($j \le m$) but are re-assigned to control ($\sigma_j = 0$).

If an assignment $\bsigma$ moves $\Delta := m - \sum_{j=1}^m \sigma_j$ of the units originally assigned to treatment to the control group (and vice-versa), then the difference between the test statistic when the shift is zero and the test statistic when the shift is $\eta$ is
\begin{equation}
    T_\eta((\bw, \bsigma)) - T_0((\bw, \bsigma)) = \eta \cdot \Delta \left ( \frac{1}{n-m} + \frac{1}{m}
    \right ). \label{eq:two_sample_shift}
\end{equation}
This result, implicit in \citet{pitman1937},
is straightforward to show:
Suppose that for a particular treatment assignment, the responses of the treatment group are $\{w_i\}_{i = 1}^m$ and the responses of the control group are $\{w_i\}_{i = m+1}^{n}$. 
Then, the difference in means can be written 
$$
\frac{(n-m) \sum_{i = 1}^m w_i - m \sum_{i = m+1}^n w_i}{m(n-m)}.
$$
For each of the $\Delta$ units that switch from treatment to control (and vice-versa) the change in
the test statistic is $\frac{(n-m)\eta - m(-\eta)}{m(n-m)} = \eta ( \frac{1}{m} + \frac{1}{n-m})$.
If $\Delta$ units change their treatment assignment, the total change in the test statistic is given by equation~\ref{eq:two_sample_shift}.

Consider $P$ as defined in equation~\ref{eq:perm-P-1}
using $T_\eta$.
It is straightforward to show that the $P$-value for fixed $\bw$ and a fixed set of treatment assignments is monotonically increasing in $\eta$:
Let $\eta_1$, $\eta_2 \in \Theta$ with $\eta_1 \le \eta_2$. 
By equation~\ref{eq:two_sample_shift}:
$$
T_{\eta_2}((\bw, \bsigma)) = \eta_2 \cdot \Delta \left ( \frac{1}{m} + \frac{1}{n-m} \right ) + T_{0}((\bw, \bsigma)).
$$
Thus, 
$$
T_{\eta_2}((\bw, \bsigma)) - T_{\eta_1}((\bw, \bsigma)) = (\eta_2 - \eta_1) \cdot \Delta \left ( \frac{1}{m} + \frac{1}{n-m} \right ) \geq 0
$$
since $\eta_2 \geq \eta_1$. 
Because this is true for every $\bsigma$, when the $P$-value is defined as in equation~\ref{eq:perm-P-1}, $P_{\eta_2}(x, u) \geq P_{\eta_1}(x, u)$. 
A similar argument applies to the other $P$-values defined in Section~\ref{sec:permute}, so the algorithms from Section~\ref{sec:confsets} apply.

As in the one-sample problem, for each assignment $\bsigma$,
we only need to keep track of the test statistic for $\eta = 0$ and the number of units that moved from treatment to control to find
the value of the test statistic for the assignment $\bsigma$ for any other value of $\eta$.
Section~\ref{sec:numerical-comparisons} gives numerical examples.

\section{Numerical comparison to previous methods}
\label{sec:numerical-comparisons}
The new method has time complexity $O(n)$ for the one-sample and two-sample problems, where $n$ is the number of data.
Garthwaite's algorithm also has time complexity $O(n)$, but
is approximate rather than conservative.
The number of sub-sample means (values of the shift where the $P$-value can change) considered by Tritchler's method
is exponential in $n$.\footnote{%
Enumeration of the $O(2^n)$ sub-sample means can be combined with a variant of the bisection algorithm to avoid testing them all; it is not clear whether the original implementation took
advantage of such savings.
}
\citet{tritchler1984} claims that each value can be tested with computation that is polynomial in $n$.
\citet{tritchler1984} reports that
computation times increased by almost a factor of 5 when the number of observations was doubled; see below.

Comparing the methods in 
Section~\ref{sec:extant-methods} numerically to the method proposed here is complicated for at least
three reasons:
\begin{enumerate}
    \item The method of \citet{tritchler1984} is not described completely
    \citep{bardelli2016} and there is no known software implementation,\footnote{%
    We are aware of two historical implementations, one in Fortran and one in C++; to the best of our knowledge,
    neither still exists.
    (D.\ Tritchler, personal communication, October 2021;
    N.\ Schmitz, personal communication, December 2021.)
    We attempted to implement the algorithm ourselves, but crucial details were missing from the published description.}
    so we are limited to comparisons with results already reported there.
    \item The assumptions of the ``typical set'' theorem Tritchler's method relies on are not satisfied in practice, so the intervals it produces in examples other than the published examples (where it reproduced the full group results) might have less than their nominal confidence level.
    \item The method of \citet{garthwaite1996,garthwaite2009} is approximate, not conservative, and in practice it does not converge
    reliably.
\end{enumerate}
Because of these issues, comparisons below are based on
the datasets from \citet{tritchler1984} and \citet{garthwaite1996}.

An open-source R implementation of the new algorithm is available in the \texttt{permutest} package \citep{permutest}
and a Python implementation is available at \url{https://github.com/akglazer/monte-carlo-ci}. 
The Python version uses the {\tt{cryptorandom}} library, which provides a cryptographic quality pseudo-random number generator (PRNG) based on the SHA256 cryptographic hash function (see section~\ref{sec:discuss}).
The higher quality of the {\tt{cryptorandom}} PRNG comes at a higher computational cost
than the Mersenne Twister (MT), the default in R, Python, MATLAB, and many other languages and packages.
For instance,
generating $10^7$ pseudo-random integers between 1 and 
$10^7$ using
cryptorandom takes 16s on our machine, in contrast to
0.17s for the {\tt{numpy}} implementation of MT.\footnote{%
Despite its computational cost, we
advocate using the higher-quality PRNG, especially for large problems: the MT state space is smaller than the number of permutations of 2084 items
\citep{starkOttoboni18}.
}

\paragraph*{Comparison to \citet{tritchler1984}.} 
We first compare results for the one-sample problem in \citet{tritchler1984} using Darwin's data on the differences in heights of 15~matched pairs of cross-fertilized and self-fertilized plants. 
Table~\ref{tab:trichtler} reports full-group confidence intervals generated by enumeration, the intervals reported by \citet{tritchler1984}, and intervals generated by our method using $N=10^4$ and $e=10^{-8}$.
While the intervals generated by Tritchler's method
agree with the full-group confidence intervals in this instance,
we are concerned that is not always true because of 
the violations of the assumptions of the typical value theorem.
The intervals generated by the new method are slightly longer (by 0.003) at confidence 90\% and shorter at 99\% (by 0.5) but are nonetheless conservative. 

To generate the
confidence intervals took an average of 0.23s of CPU time (across 100 replications) on a 2022 Apple MacBook Pro 
with an M1 Max chip and 64GB of unified memory,
running
macOS~12.3.

\begin{table}[ht]
    \centering
    \begin{tabular}{c|c|c|c}
   confidence level & full-group interval & 
   Tritchler & new method \\
    \hline
 90\%  &  [3.75, 38.14] & [3.75, 38.14]  & [3.857, 38.250]  \\
 95\%  & [-0.167, 41.0] & [-0.167, 41.0]  & [-0.167, 41.0]  \\
 99\%  & [-9.5, 47.0] & [-9.5, 47.0]  & [-8.80, 47.20]
    \end{tabular}
    \caption{One-sample confidence intervals at 90\%, 95\%, and 99\%, computed various ways: 
    using the full group of $2^{15}$ reflections, using the method of \citet{tritchler1984}, and using the new method with $N=10^4$ and $e=10^{-8}$. 
    Data on 15~matched pairs of plants \citep{tritchler1984}: 
    $49, -67, 8, 6, 16, 23, 28, 41, 14, 29, 56, 24, 75, 60, -48.$
    The new method produced slightly longer (by 0.003) interval at 90\% confidence and a slightly shorter (by 0.5) 99\% confidence interval, but the length depends on the seed.
    \label{tab:trichtler}}
\end{table}
\citet{tritchler1984} notes that execution time increases by about 250\% when the number of observations is increased by about 67\%.
Execution time for the new method (with $N=10^4$) only increased by about 26\%. 

Next we consider the two-sample problem example in \citet{tritchler1984}, which uses data reported in \citet{snedecor1967} on the effect of sleep on the basal metabolism of 26~college women.\footnote{%
    \citet{snedecor1967} attribute the data to a 1940 Ph.D.\ dissertation at Iowa
    State University.
    They appear to come from 
    an observational study with a hypothetical randomization rather than a randomized experiment: although students presumably were not randomly assigned to sleep different amounts of time, the statistical analysis assumes that labeling a student as having 0--6 hours versus 7+ hours of sleep would amount to a random label if sleep had no effect on metabolism.
    If this is an observational study, confounding is likely: students who sleep less than 7~hours differ from those who sleep more than 7~hours in ways other than how long they sleep.
}
The data are listed in table~\ref{tab:sleep}.
Full-group confidence intervals generated by enumeration, confidence intervals reported by \citet{tritchler1984}, and confidence intervals generated by our method using $N=10^4$ and $e=10^{-8}$ are given in table~\ref{tab:trichtler2}.
\begin{table}[!ht]
    \centering
    \begin{tabular}{ll} \hline
         7+ hours of sleep &  \begin{tabular}[c]{@{}l@{}} 35.3, 35.9, 37.2, 33.0, 31.9, 33.7, 36.0, 35.0, \\ 33.3, 33.6, 37.9, 35.6, 29.0, 33.7, 35.7 \end{tabular} \\ \hline
         0--6 hours of sleep & \begin{tabular}[c]{@{}l@{}} 32.5, 34.0, 34.4, 31.8, 35.0, 34.6, 33.5, 33.6, \\ 31.5, 33.8, 34.6 \end{tabular} \\ \hline
    \end{tabular}
    \caption{Basal metabolism (in $\mbox{kcal}/m^{2}/h$) and hours of sleep of 26~college women (source: \citet{tritchler1984}).}
    \label{tab:sleep}
\end{table}

\citet{tritchler1984} noted that CPU time increased by 458\% when the number of observations was doubled.
The new method took 0.28s (averaged across 100 iterations) to execute,\footnote{%
\label{note:mac}
2022 Apple MacBook Pro,
M1 Max chip, 64GB of unified memory, macOS 12.3.
} 
increasing 
to 0.45s (a 61\% increase) when the number of 
observations was doubled.


\begin{table}[ht]
    \centering
    \begin{tabular}{c|c|c|c}
   confidence level & full-group & 
   Tritchler &  new method \\
    \hline
 90\%  &  [-2.114, 0.386] & [-2.114, 0.386] & [-2.117, 0.380] \\
 95\%  & [-2.340, 0.650] & [-2.340, 0.650] & [-2.333, 0.643]  \\
 99\%  & [-2.814, 1.180] & [-2.814, 1.180] & [-2.850, 1.167]
    \end{tabular}
    \caption{Two-sample confidence intervals at 90\%, 95\%, and 99\% 
    computed from data in table~\ref{tab:sleep} in three ways: using the full group of $\binom{26}{11}$ reflections, using the method of \citet{tritchler1984}, and using the proposed method with $N=10^4$ and $e=10^{-8}$.
    The fourth column is based on a single seed.
    For that seed and these data, the intervals produced by the new method are slightly shorter at confidence levels 90\% and 95\% and slightly longer at level 99\%.}
    \label{tab:trichtler2}
\end{table}
The computational burden of the method of \citet{tritchler1984} becomes prohibitive quickly as the problem size $n$ grows.
Because the new method is $O(n)$, it
remains feasible for moderately large problems, and the workload can be decreased without sacrificing validity by reducing the number of replications $N$.
For example, for $N=10^4$ samples, $e=10^{-8}$,
and groups of $1,000$ units, the new method runs in approximately 22s on the laptop described above
using the {\tt{cryptorandom}} PRNG and approximately 5s using
the {\tt{numpy.random.random}} MT implementation.

\paragraph*{Comparison to Garthwaite's algorithm. } 
\citet{garthwaite1996} gives an example
involving the effect of malarial infection on the stamina of the lizard \emph{Sceloporis occidentali}. 
The data are listed in table~\ref{tab:lizards}.
\begin{table}[!ht]
    \centering
    \begin{tabular}{ll} \hline
         infected lizards ($X$) &  \begin{tabular}[c]{@{}l@{}}16.4, 29.4, 37.1, 23.0, 24.1, 24.5, 16.4, 29.1, \\ 36.7, 28.7, 30.2, 21.8, 37.1, 20.3, 28.3\end{tabular} \\ \hline
         uninfected lizards ($Y$) & \begin{tabular}[c]{@{}l@{}}22.2, 34.8, 42.1, 32.9, 26.4, 30.6, 32.9, 37.5, \\ 18.4, 27.5, 45.5, 34.0, 45.5, 24.5, 28.7\end{tabular} \\ \hline
    \end{tabular}
    \caption{Meters run in two minutes by infected and uninfected lizards 
    (source: \citet{samuels2003}, as reported by \citet{garthwaite1996}).}
    \label{tab:lizards}
\end{table}
\citet{garthwaite1996} reports $[-0.30, 10.69]$ as
the nominal 95\% confidence interval based on 6000~steps of the algorithm; as mentioned above,
that interval is approximate, not conservative.
Our implementation of the Garthwaite algorithm
produced the approximate confidence 
interval $[-0.29, 11.0]$,
using 6000~steps and the initial values of 0 and 10 for the lower and upper endpoints.
For $N = 6000$ using the $P$-value defined in equation~\ref{eq:perm-P-1}, the new algorithm yields the conservative confidence interval 
$[-0.20, 10.875]$, $0.215$ shorter than the Garthwaite interval 
in these runs.
Across 100 iterations, Garthwaite's algorithm took 0.47 seconds to execute and our new algorithm took 0.33 seconds.
While both algorithm execute quickly, intervals produced by the new algorithm are guaranteed to be conservative, while intervals produced by the Garthwaite algorithm are not.

\section{Discussion}
\label{sec:discuss}
This paper presents general methods for constructing conservative confidence sets and intervals by inverting Monte Carlo tests based on a single set of Monte Carlo samples.
Methods that use different samples for each hypothesized parameter value in effect try to
find the roots of a function from noisy data,
necessarily resulting in an approximate solution.
In contrast, using a single sample leads to finding the roots from noiseless data, which 
in principle can be done with arbitrarily high accuracy.

We present open-source Python and R implementations of new algorithms that construct confidence sets from conservative Monte Carlo tests when the $P$-value (for fixed data and a fixed set of Monte Carlo samples) is a monotone or weakly unimodal function of the parameter.
For problems with real-valued parameters, 
if the $P$-value is monotone in the parameter, a minor modification 
of the bisection algorithm
finds a conservative one-sided confidence within any desired precision of the best possible valid confidence bound based on that $P$-value function.
If the $P$-value is quasiconcave in the parameter,
using the modified bisection algorithm twice
yields a conservative two-sided confidence interval to any desired precision.
Additional computational savings are possible for common test statistics in the one-sample and two-sample problem 
by exploiting the relationship among values of the test statistics for different values of the parameter and for different Monte Carlo samples.

\paragraph*{How many Monte Carlo samples?}
The tests and confidence intervals are 
conservative no matter how few Monte Carlo samples $N$ are used,
so computation time can be reduced by 
decreasing $N$.
However, the highest attainable non-trivial confidence
level is $N/(N+1)$.
Increasing the number of samples also reduces the variability of results from seed to seed, and tends to approximate more precisely the interval that would
be obtained by examining all possible samples, permutations, or treatment assignments, the ``full group'' or ``all randomizations''
interval.
The length of the intervals (or power of the tests) may also depend on the number of Monte Carlo
samples.

\paragraph*{The PRNG matters.}
Many common PRNGs---including the Mersenne Twister, the default PRNG in
R, Python, MATLAB, SAS, and STATA (version 14 and later)---have state spaces that are too small for 
large problems \citep{starkOttoboni18}.
Depending on the problem size, a given PRNG might not
be able to generate all samples or all permutations, much less generate them with equal probability. 
Linear congruential generators (LCGs) are especially limited; even a 128-bit LCG can generate only about 0.03\% of the possible samples of size~25 from a set of 500~items.
The Mersenne Twister cannot in principle 
generate all permutations of set of 2084 items.
It can generate less than $10^{-7}$ of the samples of size $5,000$ from a population of size $10,000$,
and therefore certainly cannot generate such random samples uniformly.

    

The choice of algorithms for generating samples or permutations from the PRNG also matters:
a common algorithm for generating a sample---assign a pseudorandom number to each item, sort on that number, and take the first $k$ items to be the sample---requires a much higher-quality PRNG than algorithms that generate the sample more directly \citep{starkOttoboni18}.
For large problems, a \emph{cryptographic quality} (not necessarily \emph{cryptographically secure}, which generally involves periodic re-seeding) PRNG may be needed.
A Python implementation of a PRNG that uses
the SHA256 cryptographic hash function is available at \url{https://github.com/statlab/cryptorandom}.\footnote{%
    See \url{https://statlab.github.io/cryptorandom/}.
    The package is on PyPI and can be installed with {\tt{pip}}.
}

\paragraph*{Where is the computational cost?}
In the one-sample and two-sample examples above, the bulk of the CPU time is in generating the Monte Carlo sample, especially when using
the
{\tt{cryptorandom}} SHA256 PRNG.
The bisection stage of the algorithm is extremely fast:
changing $e$ from $10^{-8}$ to $10^{-6}$ has a
trivial effect on run times.
As an example, consider the two-sample problem with data from \citet{tritchler1984} presented in Section~\ref{sec:numerical-comparisons}. 
Generating a single Monte Carlo sample takes about 0.30 seconds.
Calculating the permutation test $P$-value using $10^4$ pseudorandom permutations by brute force for a single shift takes about 0.10s; using the shortcuts in Section~\ref{sec:oneTwoSample} reduces the time to approximately 0.0005s. 
This example required about 20~iterations of the bisection algorithm, for both the upper and lower endpoints, to achieve the desired tolerance $e = 10^{-8}$.
While the Section~\ref{sec:oneTwoSample} shortcuts substantially reduce runtime, the biggest reduction is from a single Monte Carlo sample to test every null: generating the Monte Carlo sample is the most time-consuming part of the computation.

\paragraph*{Confidence sets for percentiles of a `treatment effect' in the two-sample problem.}
In the two-sample problem, the null hypothesis $\theta = \eta$ is that the shift for \textit{every} unit is $\eta$.
\citet{caughey2017} show that this hypothesis and the corresponding confidence set can be interpreted 
as a hypothesis test and confidence set for 
the maximum and minimum of the shift even when the shift varies from unit to
unit, and how to use similar calculations to find
confidence bounds for percentiles of the shift.

\paragraph*{Multivariate parameters and non-quasi-concave $P$-value functions.}
When $\theta$ is multivariate (e.g., $\Theta \subset \Re^\ell$), natural test statistics (e.g., the maximum difference-in-means across components for the two-sample problem or nonparametric combinations of componentwise tests \citep{pesarin2010}) do not necessarily lead to $P$-values that are quasiconcave separately in the components of the parameter, but
one can still construct a valid confidence set by inverting conservative Monte Carlo tests using a single set of samples.
The same is true for scalar parameters when the $P$-value function is not monotone or quasiconcave.
However, the resulting confidence set is then typically not a simply connected set.
Section~\ref{sec:invert} lists some computational strategies that may help find a conservative confidence set.

\paragraph*{Nuisance parameters.}
In some problems, $\theta$ is multivariate but 
the quantity of interest is a real-valued functional $f$ of $\theta$, e.g., its weighted mean, a component, a contrast, or some other linear or nonlinear functional; other aspects of $\theta$ are nuisance parameters.
There are three general strategies to obtain conservative confidence
bounds for $f(\theta)$:
\begin{itemize}
    \item 
       Define the $P$-value to be
       the maximum $P$-value over all values of the nuisance parameters that correspond to 
       the hypothesized value of the parameter of interest, in effect decomposing a composite null into a union of
       simple nulls and rejecting the composite
       iff every simple null is rejected
       \citep{dufour2006, neyman1933, ottoboni2018}.
    \item 
       Define the $P$-value for the
       composite hypothesis to be
       the maximum $P$-value over a confidence set for the nuisance parameters and adjust the $P$-value accordingly
       \citep{berger1994}.
    \item 
       Construct a confidence set for the 
       entire parameter $\theta$, then 
       define the endpoints of the confidence interval for $f(\theta)$
       as the supremum and infimum of $f$ over that confidence
       set \citep{evans2002, stark1992}.   
\end{itemize}
The general strategy of re-using the Monte Carlo sample to test different hypothesized values of the parameter works with all of these.

\paragraph*{Software.}
Open-source Python and R software implementing the algorithms presented in this paper (i.e., inverting conservative Monte Carlo tests when the $P$-value is a monotone or weakly unimodal function of the parameter) is available at \url{https://github.com/akglazer/monte-carlo-ci} and in the R package \texttt{permutest} \citep{permutest}.
The method produces conservative intervals and is more efficient than any other conservative approach we are aware of.
Furthermore, the algorithms can be used to construct confidence intervals for a variety of test statistics and randomization schemes (not just the one- and two-sample problem with sample mean and difference-in-means test statistics respectively), so long as the $P$-value is a monotone or weakly unimodal function of the parameter.
We know of only a few computational implementations of confidence intervals for the one- or two-sample problem.
The {\tt{CIPerm}} R package constructs intervals for the two-sample problem based on the algorithm in \citet{nguyen2009}.
While that implementation uses a single set of permutations, it uses a brute-force search for the endpoints, which \citet{nguyen2009} note is time-consuming.
Another R~package, {\tt{Perm.CI}}, only works for binary outcomes and uses a brute-force method.
Other implementations, such as that in \citet{caughey2017}\footnote{%
See \url{https://github.com/li-xinran/RIQITE/blob/main/R/RI_bound_20220919.R}
} 
use the function {\tt{uniroot}}, an implementation of Brent's method, which (incorrectly) assumes the $P$-value is continuous in $\eta$.

\paragraph*{Future work.} We aim to build on the strategies outlined in Section~\ref{sec:confsets} to develop more efficient algorithms for cases where the $P$-value is neither a monotone nor a weakly unimodal function of the parameter. Additionally, further research is required to address strategies for the multidimensional case. An appropriate choice of test statistic may yield a $P$-value that is quasiconcave in the parameter, but this remains an open question requiring additional investigation.

\paragraph*{Acknowledgments.}
This work was supported in part by NSF Grants DGE~2243822 and SaTC~2228884.
We are grateful to Benjamin Recht for helpful conversations
and to P.M.\ Aronow for comments on an earlier draft.

\bibliography{bib.bib}

\begin{thebibliography}{43}
\newcommand{\enquote}[1]{``#1''}
\expandafter\ifx\csname natexlab\endcsname\relax\def\natexlab#1{#1}\fi

\bibitem[{Aronow et~al.(2023)Aronow, Chang, and Lopatto}]{aronow2023}
Aronow, P., Chang, H., and Lopatto, P. (2023), \enquote{Fast computation of
  exact confidence intervals for randomized experiments with binary outcomes,}
  arXiv:2305.09906.

\bibitem[{Bardelli(2016)}]{bardelli2016}
Bardelli, C. (2016), \enquote{Nonparametric Confidence Intervals Based on
  Permutation Tests,} Ph.D. thesis, Politecnico Di Milano.

\bibitem[{Barnard(1963)}]{barnard1963}
Barnard, G. (1963), \enquote{Discussion of The Spectral Analysis of Point
  Processes,} \textit{Journal of the Royal Statistical Society Series B}, 25,
  294--294.

\bibitem[{Berger and Boos(1994)}]{berger1994}
Berger, R. and Boos, D. (1994), \enquote{$P$ values maximized over a confidence
  set for the nuisance parameter,} \textit{Journal of the American Statistical
  Association}, 89, 1012--1016.

\bibitem[{Birnbaum(1974)}]{birnbaum1974}
Birnbaum, Z. (1974), \enquote{Computers and unconventional test-statistics,} in
  \textit{Reliability and Biometry: Statistical Analysis of Lifelength}, eds.
  Proschan, F. and Serfling, R., Philadelphia: SIAM, pp. 441--458.

\bibitem[{B{\o}lviken and Skovlund(1996)}]{bolviken1996}
B{\o}lviken, E. and Skovlund, E. (1996), \enquote{Confidence intervals from
  {M}onte {C}arlo tests,} \textit{Journal of the American Statistical
  Association}, 91, 1071--1078.

\bibitem[{Caughey et~al.(2023)Caughey, Dafoe, Li, and Miratrix}]{caughey2017}
Caughey, D., Dafoe, A., Li, X., and Miratrix, L. (2023), \enquote{Randomisation
  inference beyond the sharp null: bounded null hypotheses and quantiles of
  individual treatment effects,} \textit{Journal of the Royal Statistical
  Society Series B: Statistical Methodology}, 85, 1471--1491.

\bibitem[{Davison and Hinkley(1997)}]{davison1997}
Davison, A. and Hinkley, D. (1997), \textit{Bootstrap Methods and Their
  Application}, Cambridge: Cambridge University Press.

\bibitem[{Dufour(2006)}]{dufour2006}
Dufour, J.-M. (2006), \enquote{Monte Carlo tests with nuisance parameters: A
  general approach to finite-sample inference and nonstandard asymptotics,}
  \textit{Journal of Econometrics}, 133, 443--477.

\bibitem[{Dwass(1957)}]{dwass1957}
Dwass, M. (1957), \enquote{Modified randomization tests for nonparametric
  hypotheses,} \textit{The Annals of Mathematical Statistics}, 28, 181--187.

\bibitem[{Ernst(2004)}]{ernst2004permutation}
Ernst, M. (2004), \enquote{Permutation methods: a basis for exact inference,}
  \textit{Statistical Science}, 19, 676--685.

\bibitem[{Evans and Stark(2002)}]{evans2002}
Evans, S. and Stark, P. (2002), \enquote{Inverse problems as statistics,}
  \textit{Inverse Problems}, 18, R55--R97.

\bibitem[{Foutz(1980)}]{foutz1980}
Foutz, R.~V. (1980), \enquote{A method for constructing exact tests from test
  statistics that have unknown null distributions,} \textit{Journal of
  Statistical Computation and Simulation}, 10, 187--193.

\bibitem[{Garthwaite(1996)}]{garthwaite1996}
Garthwaite, P. (1996), \enquote{Confidence intervals from randomization tests,}
  \textit{Biometrics}, 52, 1387--1393.

\bibitem[{Garthwaite and Jones(2009)}]{garthwaite2009}
Garthwaite, P. and Jones, M. (2009), \enquote{A stochastic approximation method
  and its application to confidence intervals,} \textit{Journal of
  Computational and Graphical Statistics}, 18, 184--200.

\bibitem[{Glazer(2024)}]{permutest}
Glazer, A. (2024), \textit{permutest: Run Permutation Tests and Construct
  Associated Confidence Intervals}, {R package version 1.0.0}.

\bibitem[{Good(2006)}]{good2006}
Good, P. (2006), \textit{Resampling Methods}, Boston: Birkh\"{a}user", 3rd ed.

\bibitem[{Harrison(2012)}]{harrison12}
Harrison, M. (2012), \enquote{Conservative hypothesis tests and confidence
  intervals using importance sampling,} \textit{Biometrika}, 99, 57--69.

\bibitem[{Hartigan(1969)}]{hartigan1969}
Hartigan, J. (1969), \enquote{Using subsample values as typical values,}
  \textit{Journal of the American Statistical Association}, 64, 1303--1317.

\bibitem[{Hemerik(2024)}]{hemerik2024}
Hemerik, J. (2024), \enquote{On the term `randomization test',} \textit{The
  American Statistician}, 78, 327--334.

\bibitem[{Hemerik and Goeman(2021)}]{hemerik2021}
Hemerik, J. and Goeman, J. (2021), \enquote{Another look at The Lady Tasting
  Tea and differences between permutation tests and randomisation tests,}
  \textit{International Statistical Review}, 89, 367--381.

\bibitem[{Higgins(2004)}]{higgins2004}
Higgins, J. (2004), \textit{Introduction to Modern Nonparametric Statistics},
  Pacific Grove, CA: Brooks/Cole.

\bibitem[{Imbens and Rubin(2015)}]{imbens2015causal}
Imbens, G. and Rubin, D. (2015), \textit{Causal Inference in Statistics,
  Social, and Biomedical Sciences}, New York: Cambridge University Press.

\bibitem[{Kempthorne and Doerfler(1969)}]{kempthorne1969}
Kempthorne, O. and Doerfler, T. (1969), \enquote{The behaviour of some
  significance tests under experimental randomization,} \textit{Biometrika},
  56, 231--248.

\bibitem[{Lehmann et~al.(2005)Lehmann, Romano, and Casella}]{lehmann2005}
Lehmann, E., Romano, J., and Casella, G. (2005), \textit{Testing Statistical
  Hypotheses}, New York: Springer, 3rd ed.

\bibitem[{Neyman and Pearson(1933)}]{neyman1933}
Neyman, J. and Pearson, E. (1933), \enquote{On the problem of the most
  efficient tests of statistical hypotheses,} \textit{Philosophical
  Transactions of the Royal Society of London. Series A, Containing Papers of a
  Mathematical or Physical Character}, 231, 289--337.

\bibitem[{Nguyen(2009)}]{nguyen2009}
Nguyen, M. (2009), \enquote{Nonparametric Inference Using Randomization and
  Permutation Reference Distribution and Their MonteCarlo Approximation,} Ph.D.
  thesis, Portland State University.

\bibitem[{O'Gorman(2014)}]{ogorman2014}
O'Gorman, T. (2014), \enquote{Regaining confidence in confidence intervals for
  the mean treatment effect,} \textit{Statistics in medicine}, 33, 3859--3868.

\bibitem[{Ottoboni et~al.(2018)Ottoboni, Stark, Lindeman, and
  McBurnett}]{ottoboni2018}
Ottoboni, K., Stark, P., Lindeman, M., and McBurnett, N. (2018),
  \enquote{Risk-limiting audits by stratified union-intersection tests of
  elections ({SUITE}),} in \textit{International Joint Conference on Electronic
  Voting}, Springer, pp. 174--188.

\bibitem[{Pagano and Tritchler(1983)}]{pagano1983}
Pagano, M. and Tritchler, D. (1983), \enquote{On obtaining permutation
  distributions in polynomial time,} \textit{Journal of the American
  Statistical Association}, 78, 435--440.

\bibitem[{P.B. and Ottoboni(2018)}]{starkOttoboni18}
P.B. and Ottoboni, K. (2018), \enquote{Random sampling: practice makes
  imperfect,} arXiv:1810.10985.

\bibitem[{Pesarin and Salmaso(2010)}]{pesarin2010}
Pesarin, F. and Salmaso, L. (2010), \textit{Permutation Tests for Complex Data:
  Theory, Applications and Software}, Chichester, United Kingdom: John Wiley \&
  Sons.

\bibitem[{Phipson and Smyth(2010)}]{phipson2010}
Phipson, B. and Smyth, G. (2010), \enquote{Permutation P-values should never be
  zero: calculating exact P-values when permutations are randomly drawn,}
  \textit{Statistical Applications in Genetics and Molecular Biology}, 9.

\bibitem[{Pitman(1937)}]{pitman1937}
Pitman, E. (1937), \enquote{Significance tests which may be applied to samples
  from any populations,} \textit{Supplement to the {J}ournal of the {R}oyal
  {S}tatistical {S}ociety}, 4, 119--130.

\bibitem[{Ramdas et~al.(2023)Ramdas, Barber, Cand{\`e}s, and
  Tibshirani}]{ramdas2023}
Ramdas, A., Barber, R., Cand{\`e}s, E., and Tibshirani, R. (2023),
  \enquote{Permutation tests using arbitrary permutation distributions,}
  \textit{Sankhya A}, 85, 1--22.

\bibitem[{Robbins and Monro(1951)}]{robbinsMonro51}
Robbins, H. and Monro, S. (1951), \enquote{A stochastic approximation method,}
  \textit{The Annals of Mathematical Statistics}, 22, 400--407.

\bibitem[{Samuels et~al.(2003)Samuels, Witmer, and Schaffner}]{samuels2003}
Samuels, M., Witmer, J., and Schaffner, A. (2003), \textit{Statistics for the
  Life Sciences}, Saddle River, NJ: Prentice Hall, 4th ed.

\bibitem[{Snedecor and Cochran(1967)}]{snedecor1967}
Snedecor, G. and Cochran, W. (1967), \textit{Statistical Methods}, vol.~6,
  Ames, Iowa: Iowa State University Press.

\bibitem[{Splawa-Neyman et~al.(1990)Splawa-Neyman, Dabrowska, and
  Speed}]{neyman1990}
Splawa-Neyman, J., Dabrowska, D., and Speed, T. (1990), \enquote{On the
  application of probability theory to agricultural experiments. Essay on
  principles. Section 9.} \textit{Statistical Science}, 5, 465--472.

\bibitem[{Stark(1992)}]{stark1992}
Stark, P. (1992), \enquote{Inference in infinite-dimensional inverse problems:
  discretization and duality,} \textit{Journal of Geophysical Research: Solid
  Earth}, 97, 14055--14082.

\bibitem[{Tritchler(1984)}]{tritchler1984}
Tritchler, D. (1984), \enquote{On inverting permutation tests,} \textit{Journal
  of the American Statistical Association}, 79, 200--207.

\bibitem[{Wu and Ding(2021)}]{wu2021}
Wu, J. and Ding, P. (2021), \enquote{Randomization tests for weak null
  hypotheses in randomized experiments,} \textit{Journal of the American
  Statistical Association}, 116, 1898--1913.

\bibitem[{Zhang and Zhao(2023)}]{zhang2023}
Zhang, Y. and Zhao, Q. (2023), \enquote{What is a randomization test?}
  \textit{Journal of the American Statistical Association}, 118, 2928--2942.

\end{thebibliography}

\end{document}